\documentclass[11pt]{article}
\usepackage[a4paper,lmargin=1in,rmargin=1in]{geometry}
\usepackage{amsmath}
\usepackage{amsfonts}
\usepackage{amsbsy}
\usepackage{amsthm}
\usepackage{accents}
\usepackage{tensor}
\usepackage{dcolumn}
\usepackage{bm}
\newtheorem{thrm}{Theorem}[section]
\newtheorem{crlr}[thrm]{Corollary}
\newtheorem{lmm}[thrm]{Lemma}

\theoremstyle{definition}
\newtheorem{rmrk}{Remark}

\DeclareMathOperator{\End}{End}
\DeclareMathOperator{\id}{id}

\newcommand{\rmi}{\mathrm{i}}
\newcommand{\rmd}{\mathrm{d}}
\newcommand{\smiled}[1]{\accentset{\smile}{#1}}
\newcommand{\sympman}{\mathcal{M}}
\newcommand{\bund}{\mathcal{E}}
\newcommand{\ebund}{\End(\bund)}
\newcommand{\wbund}{\mathcal{W}}
\newcommand{\connsymp}{\partial^{S}}
\newcommand{\connbund}{\partial^{\bund}}
\newcommand{\connend}{\partial^{\ebund}}
\newcommand{\gambund}{\Gamma^{\bund}}

\newcommand{\curvbund}{R^{\bund}}
\newcommand{\ghmtp}[4]{G^{#1 ,\, #2}_{#3 \, #4}}
\newcommand{\dotghmtp}[4]{\dot{G}^{#1 ,\, #2}_{#3 \, #4}}
\newcommand{\ttoper}[3]{{T_t}^{#1 ,\, #2}_{#3}}
\newcommand{\ttinvoper}[3]{{T^{-1}_t}^{#1 ,\, #2}_{#3}}
\newcommand{\toper}[3]{T^{#1 ,\, #2}_{#3}}
\newcommand{\tinvoper}[3]{{T^{-1}}^{#1 ,\, #2}_{#3}}

\newcommand{\uoper}[3]{{U}^{#1 ,\, #2}_{#3}}
\newcommand{\uinvoper}[3]{{U^{-1}}^{#1 ,\, #2}_{#3}}
\newcommand{\moper}[2]{M^{#1}_{#2}}
\newcommand{\frmbr}[1]{\langle #1 \rangle}
\begin{document}
\linespread{1.3}

\title{Seiberg-Witten equations from Fedosov \\ deformation quantization of endomorphism bundle}
\date{}
\author{Micha{\l} Dobrski
\footnote{michal.dobrski@p.lodz.pl}
\\
\small
\emph{Centre of Mathematics and Physics}
\\
\small
\emph{Technical University of {\L}\'od\'z,}
\\
\small
\emph{Al.~Politechniki 11, 90-924 {\L}\'od\'z, Poland}}
\maketitle
\abstract{It is shown how Seiberg-Witten equations can be obtained by means of Fedosov deformation quantization of endomorphism bundle and the corresponding theory of equivalences of star products. In such setting, Seiberg-Witten map can be iteratively computed for arbitrary gauge group up to any given degree with recursive methods of Fedosov construction. Presented approach can be also considered as a generalization of Seiberg-Witten equations to Fedosov type of noncommutativity.}

\section{Introduction}
Seiberg-Witten equations define correspondence between commutative and noncommutative gauge theories \cite{seibwitt}. The noncommutativity considered in this context is given by Moyal star product. Solutions of these equations were found, analyzed and used in noncommutative field theories by many authors (compare e.g. \cite{asakawa1, goto, jurco0, brace, kraus, moller,ohl,trampetic,ulker} and references therein). Original formulation of \cite{seibwitt} was based on the relation between two possible regularizations of the particular string theory model. However, it was shown that Seiberg-Witten maps can be consistently introduced without referring to the string theory, but only by algebraic analysis of postulated noncommutative gauge transformations \cite{jurco0}. Another possibility (which will be used in our approach) is given by the theory of equivalence of star products \cite{jurco1}. 

In this paper, we are going to show that Seiberg-Witten equations and their solutions can be understood as a local manifestation of the global theory -- deformation quantization of endomorphism bundle \cite{fedosov,fedosov0}. The main inspiration for such considerations comes from the work of Jur\v{c}o and Schupp showing that the star equivalence relations can yield Seiberg-Witten equations \cite{jurco1}. The approach of \cite{jurco1}, with Kontsevich theory as a main tool, seems to work without any problems only in the case of abelian gauge. (The nonabelian gauge was analyzed in \cite{jurco2}. Unfortunately the theory became significantly complicated and the central result was stated without strict proof in this case). The main aim of the present paper is to show that equivalence theory in Fedosov quantization yields results analogous to that of \cite{jurco1}, but for arbitrary gauge group. Moreover, Fedosov machinery contains iterative procedures which can be used for explicit computation of solutions of Seiberg-Witten equations up to an arbitrary order. Also, our approach gives some generalization of these equations to Fedosov type of noncommutativity.

The paper is organized as follows. First (section 2), we recall Fedosov formalism of deformation quantization of endomorphism bundle. Next (section 3), the theory of isomorphisms generated by Heisenberg equation is used to derive the central relation of this paper -- formula (\ref{triv_ebund_covar}). We use the same theory to introduce commuting derivations of Fedosov algebra (following \cite{dobrski} at this point). Then (section 4), the things are put together and it is shown how Seiberg-Witten equations (and their solutions) can be obtained from Fedosov equivalence theory. Some concluding remarks are given in section 5.
\section{Fedosov Construction of Deformation Quantization of Endomorphism Bundle}
Let us recall Fedosov construction for deformation quantization of endomorphism bundle. We strictly follow sections {5.2} and {5.3} of \cite{fedosov}. Instead of considering general formalism in which one deals with arbitrary symplectic bundle, we stick to the simplest case of symplectic manifold $(\sympman,\omega)$ which is the base manifold for considered bundles. Since this section is given mainly for the purpose of fixing notations, the proofs are omitted and the numbers of theorems originally formulated in \cite{fedosov} are quoted. For detailed
insight into geometrical ideas behind Fedosov construction (in the case of functions on symplectic manifold) one may refer to \cite{emmrwein}. Some further properties and examples can be found in \cite{tosiek,tosiek2}.

The starting point for the Fedosov construction is a Fedosov manifold $(\sympman,\omega,\connsymp)$ i.e. a symplectic manifold $(\sympman,\omega)$ with some symplectic (torsionless and preserving $\omega$) connection $\connsymp$ (compare \cite{gelfand,bielgutt} for details). Then one considers a complex vector bundle $\bund$ over $\sympman$ with a connection $\connbund$. Our main interest will be focused on the bundle $\ebund$ which fiber $\ebund_x$ is a vector space of endomorphisms of corresponding fiber $\bund_x$. Recall that the connection $\connbund$ induces a connection $\connend$. Indeed, if in the local frame $e$ of $\bund$ one has $\connbund a=\rmd a+\gambund a$  then  $\connend B=\rmd B+[\gambund,B]$, where $\gambund$ is a local, endomorphism-valued one-form. Sections of $\ebund$ can be multiplied with usual composition of linear mappings. Fedosov construction (usually considered as a deformation quantization of algebra of functions on a symplectic manifold) yields fully geometrical, formal deformation of product of endomorphisms. Locally this construction can be understood as a deformation of matrix product.

In the first step, one introduces the bundle $W \otimes \ebund$ on the base manifold $\sympman$, where $W$ is the formal Weyl algebra bundle of usual Fedosov construction (\cite{fedosov} section {5.1}). $W \otimes \ebund$ is called the formal Weyl algebra bundle with twisted coefficients and will be denoted by $\wbund$. Its fibers are algebras $\wbund_x$ consisting of formal mappings from $T_x\sympman$ to $\ebund_x[[h]]$ of the form
\begin{equation}
\label{fedo_fps}
a(y)=\sum_{k,p \geq 0} h^k a_{i_1 \dots i_p}y^{i_1} \dots y^{i_p},
\end{equation}
where $y \in T_x\sympman$, $a_{i_1 \dots i_p}$ are components of some symmetric, $\ebund_x$-valued covariant tensors in local coordinates and $h$ is a formal parameter. One prescribes degrees to monomials in formal sum (\ref{fedo_fps}) according to the rule
\begin{equation*}
\deg (h^k a_{i_1 \dots i_p}y^{i_1} \dots y^{i_p} )=2k+p.
\end{equation*}
For nonhomogeneous $a$ its degree is given by the lowest degree of nonzero monomials in formal sum (\ref{fedo_fps}). The operator $P_m$ extracts monomials of degree $m$ from given $a$
\begin{equation*}
P_m(a)(y)=\sum_{2k+p=m} h^k a_{i_1 \dots i_p}y^{i_1} \dots y^{i_p}.
\end{equation*}

The fiberwise $\circ$-product is defined by the Moyal formula
\begin{equation*}
a \circ b = \sum_{m=0}^{\infty}\left( -\frac{\rmi h}{2}\right)^m \frac{1}{m!} 
\frac{\partial^m a}{\partial y^{i_1} \dots \partial y^{i_m}}
\omega^{i_1 j_1} \dots \omega^{i_m j_m}
\frac{\partial^m b}{\partial y^{j_1} \dots \partial y^{j_m}}.
\end{equation*}
(We follow conventions of \cite{fedosov} with minus sign at the imaginary unit. The transformation to the most popular form of Moyal product can be easily obtained by putting $h \to -h$).	
This definition is invariant under linear transformations of $y^i$ generated by transitions between local coordinates on $\sympman$. Notice that $\circ$-product is expressed in terms of the product of endomorphisms which is noncommutative from the very beginning.

We also consider the bundle $\wbund \otimes \Lambda$. Sections of this bundle can be locally written as
\begin{equation*}
a=\sum h^k a_{i_1 \dots i_p j_1 \dots j_q}(x) y^{i_1} \dots y^{i_p} \rmd x^{j_1} \wedge \dots \wedge \rmd x^{j_q}.
\end{equation*}
The $\circ$-product in $\wbund \otimes \Lambda$ is defined by the rule $(a \otimes \eta) \circ (b \otimes \xi)=(a \circ b) \otimes (\eta \wedge \xi)$. The commutator of $a \in \wbund \otimes \Lambda^r$ and $b \in \wbund \otimes \Lambda^s$ is given by $[a,b]=a \circ b - (-1)^{rs} b \circ a$. It can be easily observed that the only elements of $\wbund \otimes \Lambda$ which vanish on all commutators are these which are proportional to the identity endomorphism (or equal to $0$) at each point and which contain no $y^i$. Thus, \emph{scalar forms} are just sections belonging to $C^{\infty}(\Lambda)[[h]]$. In particular, scalar $0$-forms can be identified with formal power series with coefficients being the functions on $\sympman$.

Notice that (comparing to the case of Fedosov construction for functions) one deals with one subtlety related to initial noncommutativity of endomorphism product. When operators of the form $K=\frac{\rmi}{h} [s,\cdot \,]$ are considered, one may not take arbitrary $s$, since in general negative powers of $h$ may be produced. The only allowed $s \in C^{\infty}(\wbund \otimes \Lambda)$ are these for which monomials of the form $h^0 a_{i_1 \dots i_p j_1 \dots j_q} y^{i_1} \dots y^{i_p} \rmd x^{j_1} \wedge \dots \wedge \rmd x^{j_q}$ are expressed by central endomorphisms only. Let us call them \emph{$\mathcal{C}$-sections}. Also, let us use the term  \emph{$\mathcal{C}$-operator} for mappings which transport  $\mathcal{C}$-sections to $\mathcal{C}$-sections.
\begin{lmm}
\label{fedo_kop_lmm}
For arbitrary $\mathcal{C}$-section $s$ the commutator $\frac{i}{h}[s,\cdot \,]$ is a $\mathcal{C}$-operator.
\end{lmm}
The proof is straightforward.

One introduces an operator $\delta$ acting on elements of $\wbund \otimes \Lambda$ by the relation
\begin{equation*}
\delta a=\rmd x^k \wedge \frac{\partial a}{\partial y^k}=-\frac{\rmi}{h}[\omega_{ij} y^i \rmd x^j,a].
\end{equation*}
Similarly, $\delta^{-1}$ acting on monomial $a_{km}$ with $k$-fold $y$ and $m$-fold $\rmd x$ yields
\begin{equation*}
\delta^{-1}a_{km}=\frac{1}{k+m}y^s \iota \left(\frac{\partial}{\partial x^s}\right) a_{km}
\end{equation*}
for $k+m>0$ and $\delta^{-1}a_{00}=0$. Both $\delta$ and $\delta^{-1}$ are nilpotent and for $a \in \wbund \otimes \Lambda^k$ the Leibniz rule $\delta(a \circ b)= (\delta a) \circ b + (-1)^k a \circ \delta b$ holds. An arbitrary  $a \in \wbund \otimes \Lambda$ can be decomposed into
\begin{equation*}
a=a_{00}+\delta \delta^{-1}a + \delta^{-1}\delta a.
\end{equation*}

Symplectic connection $\connsymp$ induces the connection in formal Weyl algebras bundle $W$ which will be denoted by the same symbol. In the case of twisted coefficients (bundle $\wbund$) one deals with connection $\partial=\connsymp \otimes 1 + 1 \otimes \connend$. Let $a$ be a section of $\wbund$.
Using Darboux coordinates on $\sympman$ and symplectic connection coefficients $\tensor{\Gamma}{^l_{jk}}$ one may write $\partial a$ (for some fixed frame in $\bund$) in the form
\begin{equation*}
\partial a = \rmd a + \frac{\rmi}{h}[1/2\, \Gamma_{ijk}y^i y^j \rmd x^k,a]+[\gambund , a],
\end{equation*}
with $\Gamma_{ijk}=\omega_{il} \tensor{\Gamma}{^l_{jl}}$ (any further raising or lowering of indices is also performed by means of $\omega$).  When dealing with sections of $\wbund \otimes \Lambda$, we can compute $\partial$ using the rule
\begin{equation*}
\partial (\eta \circ a) = \rmd \eta \circ a + (-1)^k \eta \circ \partial a,
\end{equation*}
for $\eta$ being a scalar $k$-form. The $\circ$-Leibniz rule holds for $\partial$ and one could be interested in other connections with this property, namely in the connections of the form
\begin{equation*}
\nabla=\partial + \frac{\rmi}{h}[\gamma,\cdot\,],
\end{equation*}
with $\gamma \in C^{\infty}(\wbund \otimes \Lambda^1)$. One can calculate that $\nabla^2 = \frac{\rmi}{h}[\Omega, \cdot \,]$ with the curvature $2$-form $\Omega=R+\partial \gamma+\frac{\rmi}{h}\gamma \circ \gamma$ where $R=1/4\, R_{ijkl}y^i y^j \rmd x^k \wedge \rmd x^l- \frac{\rmi h}{2}\ \curvbund_{kl} \rmd x^k \wedge \rmd x^l$. Here $\tensor{R}{^i_{jkl}}$ denotes the curvature tensor of symplectic connection and $\curvbund_{kl}={\partial \gambund_l} / {\partial x_k}-{\partial \gambund_k} / {\partial x_l}+[\gambund_k,\gambund_l]$.
The connection $D$ is called Abelian if it is flat ($D^2=0$) i.e.\ if its curvature is a scalar form. The following theorem holds.
\begin{thrm}[Fedosov {5.3.3}]
\label{fedo_abel}
For arbitrary connection $\partial$ and arbitrary $\mathcal{C}$-section $\mu \in C^{\infty}(\wbund)$ with $\deg \mu \geq 3$ and $\mu|_{y=0}=0$ there exists unique Abelian connection
\begin{equation*}
D=-\delta+\partial+\frac{\rmi}{h}[r,\cdot\,],
\end{equation*}
satisfying the conditions: 
\begin{itemize}
\item $\Omega=-1/2\, \omega_{ij}\rmd x^i \wedge \rmd x^j$  is the corresponding curvature $2$-form,
\item $1$-form $r$ is $\mathcal{C}$-section,
\item $\delta^{-1}r=\mu$, 
\item $\deg r \geq 2$.
\end{itemize}
The $1$-form $r$ is the unique solution of the equation 
\begin{equation}
\label{fedo_abeliter}
r=r_0 + \delta^{-1}\left( \partial r + \frac{\rmi}{h}r \circ r\right).
\end{equation}
with $r_0=\delta^{-1}R + \delta \mu$.
\end{thrm}
Omitting the proof we remark that the solution of (\ref{fedo_abeliter}) comes from recursive application of $\id + \delta^{-1}\left( \partial  + \frac{\rmi}{h} (\cdot)^{2_{\circ}}  \right)$ which for $1$-forms is an $\mathcal{C}$-operator (here $(a)^{2_{\circ}}=a \circ a$). The initial point for this recurrence is given by the $\mathcal{C}$-section $r_0$. By lemma \ref{fedo_kop_lmm} we infer that $D$ is a $\mathcal{C}$-operator.

Section $a \in C^{\infty}(\wbund)$ is called flat if $Da=0$. Flat sections form subalgebra of the algebra of all sections of $\wbund$. We denote this subalgebra by $\wbund_D$. For $a \in C^{\infty}(W \otimes \Lambda)$ define $Q(a)$ as a solution of the equation
\begin{equation}
\label{fedo_q_iter}
b=a + \delta^{-1}(D+\delta)b
\end{equation}
with respect to $b$. One can prove that this solution is unique, and that $Q$ is a linear bijection. Then, $Q^{-1}a=a-\delta^{-1}(D+\delta)a$. The following theorem enables construction of star product of endomorphisms.
\begin{thrm}[Fedosov {5.2.4}]
\label{fedo_qbiject}
The mapping $Q$ establishes bijection between sections from $C^{\infty}(\ebund)[[h]]$ and $\wbund_D$.
\end{thrm}
Notice that for $a \in \wbund_D$ one obtains $Q^{-1}(a)=a_0$, where $a_0=a|_{y=0}$. We may also infer that $Q$ is an $\mathcal{C}$-operator. This comes from the fact that $Q$ is expressed in terms of $\mathcal{C}$-operator $\delta^{-1}(D+\delta)$. For $A,B \in C^{\infty}(\ebund)[[h]]$ the star product is defined according to the rule
\begin{equation*}
A*B=Q^{-1}(Q(A) \circ Q(B)).
\end{equation*}
The star product obtained in this way is a globally defined associative deformation of the usual product in $C^{\infty}(\ebund)$.
\begin{rmrk}
In proofs of theorems \ref{fedo_abel} and \ref{fedo_qbiject} the \emph{iteration method} is used. Given equation of the form 
\begin{equation}
\label{fedo_itermet}
a=b +K(a)
\end{equation} 
one may try to solve it iteratively with respect to $a$, by putting $a^{(0)}=b$ and $a^{(n)}=b+K(a^{(n-1)})$. If $K$ is linear and \emph{raises degrees} (i.e.\ $\deg a < \deg K(a)$ or $K(a)=0$) then it can be easily deduced that the unique solution of (\ref{fedo_itermet}) is given by the series of relations $P_m(a)=P_m(a^{(m)})$. In the case of nonlinear $K$ (as in (\ref{fedo_abeliter})) the more careful analysis is required \cite{fedosov}.
\end{rmrk}
\begin{rmrk}
All structures defined above are covariant with respect to both frame and coordinate transformations. Let us consider the first covariance. Suppose we have some section $a$ of $\wbund$. Locally, in frame $e$, $a$ is represented by the matrix $a_{\frmbr{e}}$ with entries in $W \otimes \Lambda$. (We are going to adopt the following convention. For matrices and matrix-valued operators which are frame-dependent representations of global objects we put subscript ${}_{\frmbr{e}}$ to mark the corresponding frame. However, we will omit this marking if the frame is evident within the given context). When we perform transformation $\widetilde{e} = eg^{-1}$ this matrix transforms according to $a_{\frmbr{\widetilde{e}}} = g a_{\frmbr{e}} g^{-1}=g \circ a_{\frmbr{e}} \circ g^{-1}$. One easily checks that this relation holds for $r$. Hence, if
\begin{equation*}
D_{\frmbr{e}}=-\delta+\partial+\frac{\rmi}{h}[r_{\frmbr{e}},\cdot\,]
\end{equation*}
and
\begin{equation*}
D_{\frmbr{\widetilde{e}}}=-\delta+\partial+\frac{\rmi}{h}[r_{\frmbr{\widetilde{e}}},\cdot\,]
\end{equation*}
then
\begin{equation*}
D_{\frmbr{\widetilde{e}}}(g \circ a_{\frmbr{e}} \circ g^{-1})=
g \circ D_{\frmbr{e}}(a_{\frmbr{e}})) \circ g^{-1}
\end{equation*}
and in consequence
\begin{equation*}
Q_{\frmbr{\widetilde{e}}}(g \circ a_{\frmbr{e}} \circ g^{-1})=
g \circ Q_{\frmbr{e}}(a_{\frmbr{e}})) \circ g^{-1}.
\end{equation*}
\end{rmrk}
\begin{rmrk}
One recovers usual Fedosov construction of star product of functions on $\sympman$ by considering bundle $\bund_0=\sympman \times \mathbb{C}$ with global frame $e_1=1$ and global connection $1$-form $\Gamma^{\bund_0} \equiv 0$.
\end{rmrk}
\begin{rmrk}
Let both connections $\connsymp$ and $\connbund$ be flat and $\mu \equiv 0$. Suppose we fix Darboux coordinates and frame in $\bund$ with $\Gamma_{ijk} \equiv 0$ and $\gambund \equiv 0$. Then Abelian connection reads $D_T=\rmd-\delta$. The corresponding subalgebra of flat sections is called \emph{trivial algebra} and will be denoted by $\wbund_{D_T}$. The star product obtained in this case is the Moyal one.
\end{rmrk}
\begin{rmrk}
\label{fedo_rmrk_flatbund}
Also, let us consider the case of arbitrary $\connsymp$, flat $\connbund$ and $\mu \equiv 0$. Then iteration for $r$ starts with section of $\wbund$ containing only central endomorphisms. Applying $\id + \delta^{-1}\left( \partial  + \frac{\rmi}{h} (\cdot)^{2_{\circ}}  \right)$ we do not produce any noncentral endomorphisms. By inductive use of this argument we may infer that the section $r$ obtained as the solution of (\ref{fedo_abeliter}) is exactly identical to the $r$ obtained by the Fedosov construction of deformation quantization for functions on $\sympman$. Thus locally, in a frame $e$ with $\gambund \equiv 0$, we observe that $b_{\frmbr{e}}=Q_{\frmbr{e}}(a_{\frmbr{e}})$ is a matrix valued section with entries $b_{\frmbr{e}\,ij}=Q_S(a_{\frmbr{e}\,ij})$ where $Q_S$ is the bijective mapping which generate Fedosov star product of functions. Hence, in such a frame, the resulting star product of endomorphisms takes form of the usual matrix product with commutative multiplication of entries replaced by the Fedosov star product of functions on $\sympman$. We denote the corresponding Abelian connection by $D_S=\rmd+\frac{\rmi}{h}[\gamma_S,\cdot \,]$, the algebra of flat sections by $\wbund_{D_S}$ and the star product by $*_S$.
\end{rmrk}
\begin{rmrk}
The bundle $\wbund$ can be generalized in the following way. Let $\wbund^{+}$ denote the bundle with fiber $\wbund^{+}_x$ given by elements of the form
\begin{equation*}
a(y)=\sum_{2k+p \geq 0} h^k a_{i_1 \dots i_p}y^{i_1} \dots y^{i^p},
\end{equation*}
where the number of monomials with given degree $2k+p$ is finite. In other words, we admit negative powers of $h$ as long as they are compensated by positive powers of $y$ and $P_m(a)$ contains finite number of monomials for arbitrary $m \geq 0$. The $\circ$-product, connections (including $D$) and operator $Q$ extends in a natural way to $\wbund^{+}$ and $\wbund^{+} \otimes \Lambda$.
\end{rmrk}	
We need one more theorem for further purposes.
\begin{thrm}[Fedosov {5.2.6}]
\label{triv_deqlemma}
Equation $Da=b$ (for some given $b \in C^{\infty} (\wbund \otimes \Lambda^p)$, $p>0$) has a solution if and only if $Db=0$. The solution may be chosen in the form $a=-Q \delta^{-1}b$.
\end{thrm}

\section{Equivalences of Star Products Generated by Heisenberg Equation}
In this section general methods developed by Fedosov (\cite{fedosov} section 5.4) are applied to the specific case of deformation quantization of endomorphism bundle. We do not make use of full theory constructed by Fedosov. Instead we extract only those parts which are required for trivialization procedures. 

The term \emph{trivialization} will denote the procedure of establishing isomorphism between some given algebra $\wbund_D$ and $\wbund_{D_T}$ or $\wbund_{D_S}$. The construction of such isomorphisms is based on the following theorem. 
\begin{thrm}[Fedosov 5.4.3]
\label{triv_trhm_liouv}
Let 
\begin{eqnarray*}
D_t &=& -\delta + \partial^{(t)} + \frac{\rmi}{h}[r(t),\cdot \:]
\stackrel{\mathrm{locally}}{=}
-\delta + \rmd  + \frac{\rmi}{h}[1/2 \Gamma_{ijk}(t) y^{i}y^{j}\rmd x^k -\rmi h \gambund(t) + r(t),\cdot \:]\\
&=&
\rmd + \frac{\rmi}{h}[\gamma (t),\cdot \:]
\end{eqnarray*}
be a family of Abelian connections parameterized by $t \in [0,1]$, and let $H(t)$ be a $t$-dependent $\mathcal{C}$-section of $\wbund$ (called Hamiltonian) satisfying the following conditions:
\begin{enumerate}
\item $D_t H(t) - \dot{\gamma}(t)$ is a scalar form,
\item $\mathrm{deg}(H(t)) \geq 3$.
\end{enumerate}
Then, the equation
\begin{equation}
\label{triv_heis}
\frac{\rmd a}{\rmd t}+\frac{\rmi}{h}[H(t),a]=0
\end{equation}
has the unique solution $a(t)$ for any given $a(0) \in \wbund \otimes \Lambda$ and the mapping $a(0) \mapsto a(t)$ is bijective for any $t \in [0,1]$. Moreover, $a(0) \in \wbund_{D_0}$ if and only if $a(t) \in \wbund_{D_t}$.
\end{thrm}
The proof can be done by integrating the equation (\ref{triv_heis}) to
\begin{equation}
\label{triv_heis_int}
a(t)=a(0)-\frac{\rmi}{h} \int_0^t [H(\tau),a(\tau)] d \tau.
\end{equation}
and using the iteration method. Notice, that solutions of (\ref{triv_heis}) can be written in the form $a(t)=U^{-1}(t) \circ a(0) \circ U(t)$, where $U(t)$ is given as the iterative solution of
\begin{equation}
\label{triv_U_def}
U(t)=1+\frac{\rmi}{h} \int_0^t U(\tau) \circ  H(\tau) d \tau.
\end{equation}
In general, $U$ is a section of $\wbund^{+}$.

Locally one can construct isomorphism between arbitrary two  $\wbund_{D_0}$ and $\wbund_{D_1}$, both coming from theorem \ref{fedo_abel}. Indeed, fixing some scalar closed $1$-form $\lambda$ and some homotopy $\gamma(t)$ such that $\gamma(0)=\gamma_0$ and $\gamma(1)=\gamma_1$ we find $H$ as a solution to $D_t H(t) = \lambda + \dot{\gamma}(t)$. Since $D_t (\lambda + \dot{\gamma}(t))=\dot{\Omega}=0$, then $H(t)$ can be chosen as $-Q_t \delta^{-1}(\lambda(t) + \dot{\gamma}(t))$ (theorem \ref{triv_deqlemma}). As a quite direct consequence of these considerations one obtains the following theorem
\begin{thrm}[Fedosov 5.5.1]
\label{triv_trhm_triv}
Any algebra $\wbund_D$ coming from theorem \ref{fedo_abel} is locally isomorphic to the trivial algebra $W_{D_T}$.
\end{thrm}
Hence, we obtain the following useful relation.
\begin{crlr}[Fedosov 5.5.2]
\label{triv_comm_spec}
If section $a \in C^{\infty}(\wbund \otimes \Lambda)$ commutes with arbitrary $b \in \wbund_D$, where $D$ is 	generated by theorem \ref{fedo_abel}, then $a$ is necessarily a scalar form.
\end{crlr} 
The proof can be easily obtained for $\wbund_{D_T}$. Then, the result can be transported to arbitrary $\wbund_D$ by means of isomorphism mentioned in theorem \ref{triv_trhm_triv}.

Let $T_t$ denote isomorphism obtained in theorem \ref{triv_trhm_liouv}. Thus, $T_0=\id$ and $T_t (\wbund_{D_0})= \wbund_{D_t}$. Representing $T_t$ by means of (\ref{triv_U_def}) we have $T_t (a)=U^{-1}(t) \circ a \circ U(t)$. Define $\widetilde{D}_t=T_t D_0 T^{-1}_t$. One quickly calculates that
\begin{equation*}
\widetilde{D}_t=D_0 + [U^{-1}(t) \circ D_0 U(t), \cdot\,].
\end{equation*}
Consequently we may define $\widetilde{\gamma}(t)=\gamma(0)-\rmi h \,U^{-1}(t)  \circ D_0 U(t)$.
Now we are going to investigate how $\widetilde{D}_t$ and $\widetilde{\gamma}(t)$ are related to $D_t$ and $\gamma(t)$.
\begin{lmm}
\label{triv_lmm_gamma}
With above notations and under assumptions of theorem \ref{triv_trhm_liouv}, the connection $\widetilde{D}_t$ is Abelian and $\widetilde{D}_t=D_t$. If additionally $D_t H(t) - \dot{\gamma}(t)=0$ then $\widetilde{\gamma}(t)=\gamma(t)$.
\end{lmm}
\begin{proof}
The Abelian property for $\widetilde{D}_t$ comes from obvious relation $\widetilde{D}_t\widetilde{D}_t=T_t D_0 D_0 T^{-1}_t$. Using theorem \ref{triv_trhm_liouv} one quickly computes that $\wbund_{\widetilde{D}_t}=\wbund_{D_t}$. Indeed, for $a \in \wbund_{D_t}$ it follows that $\widetilde{D}_t a = T_t (D_0 T^{-1}_t(a))=T_t(0)=0$. Similarly, for $b \in \wbund_{\widetilde{D}_t}$, from $0=T_t (D_0 T^{-1}_t(b))$ one infers that $b \in \wbund_{D_t}$. Thus, for arbitrary $c \in \wbund_{D_t} = \wbund_{\widetilde{D}_t}$ we have $\widetilde{D}_t c = D_t c =0$. This yields  $[\widetilde{\gamma}(t)-\gamma(t),c]=0$. Hence, from corollary  \ref{triv_comm_spec} one obtains that $\widetilde{\gamma}(t)-\gamma(t)$ is a scalar $1$-form. But this means that $\widetilde{D}_t a = D_t a$ for arbitrary $a \in C^{\infty}(\wbund \otimes \Lambda)$.

Let us compute $\dot{\widetilde{\gamma}}$. From (\ref{triv_U_def}) it follows that 
$\dot{U}=\frac{\rmi}{h} U \circ H$ and $\dot{U}^{-1}=-\frac{\rmi}{h} H \circ U^{-1}$. Thus
\begin{eqnarray*}
\dot{\widetilde{\gamma}}&=&-\rmi h \, \dot{U}^{-1} \circ D_0 U -\rmi h \,U^{-1} \circ D_0 \dot{U}
=-H \circ U^{-1} \circ D_0 U + (U^{-1} \circ D_0 U) \circ H\\
& & +D_0 H=D_0 H +  [U^{-1} D_0  \circ U, H]=\widetilde{D}_t H = D_t H.
\end{eqnarray*}
If $D_t H - \dot{\gamma}=0$ then $\dot{\widetilde{\gamma}}=\dot{\gamma}$. But we know that $\widetilde{\gamma}(0)=\gamma(0)$. Thus $\widetilde{\gamma}(t)=\gamma(t)$.
\end{proof}

Finally, let us observe that isomorphisms $T_t$ are fully covariant with respect to transformations of the frame $\widetilde{e} = eg^{-1}$ i.e.
\begin{equation}
\label{triv_tcovar}
{T_t}_{\frmbr{\widetilde{e}}}(g \circ a_{\frmbr{\widetilde{e}}} \circ g^{-1})= g \circ {T_t}_{\frmbr{e}}(a_{\frmbr{e}}) \circ g^{-1}.
\end{equation}
\subsection{Local isomorphism between $\wbund_D$ and $\wbund_{D_S}$}
Now, we are going to analyze isomorphisms which will be the main tool for the construction of Seiberg-Witten maps.

Let us fix some local frame $e$ in $\bund$. Consider a homotopy $\ghmtp{\gambund}{0}{\frmbr{e}}{}(t)$ of connection coefficients with $\ghmtp{\gambund}{0}{\frmbr{e}}{}(0)=\gambund$ and $\ghmtp{\gambund}{0}{\frmbr{e}}{}(1)=0$. Thus, indices of $G$ describe current frame, the starting and the ending point of the homotopy. We also fix homotopy $m(t)$ of normalizing coefficients satisfying $m(0)=\mu$ and $m(1)=0$.
$\ghmtp{\gambund}{0}{\frmbr{e}}{}(t)$ together with $m(t)$ generate family of Abelian connections $D_t=\rmd + \frac{i}{h}[\gamma(t),\cdot \,]$, each coming from theorem \ref{fedo_abel}. Obviously $D_0$ is our initial connection in $\wbund$ and ${D_1}_{\frmbr{e}}=D_S$.
Using theorem \ref{triv_deqlemma} and requiring $D_t H= \dot{\gamma}$ we may find Hamiltonian $H(t)=-Q_t(\delta^{-1}\dot{\gamma})$. But $\dot{\gamma}=-\rmi h \dotghmtp{\gambund}{0}{\frmbr{e}}{}+\dot{r}$. Thus, since $\rmd / \rmd t$ commutes with $\delta^{-1}$
\begin{equation*}
\delta^{-1}\dot{\gamma}=-\rmi h \delta^{-1} \dotghmtp{\gambund}{0}{\frmbr{e}}{} + \delta^{-1}\dot{r}=-ih \dotghmtp{\gambund}{0}{\frmbr{e}}{j} y^j + \dot{m}
\end{equation*}
and consequently
\begin{equation}
\label{triv_ebund_ham}
H(t)=Q_t(\rmi h \dotghmtp{\gambund}{0}{\frmbr{e}}{j} y^j - \dot{m}).
\end{equation}
Observe that $H(t)$ is a $\mathcal{C}$-section obtained by the action of $\mathcal{C}$-operator on $\mathcal{C}$-section. Family of isomorphisms generated by this Hamiltonian will be denoted as
\begin{equation*}
\ttoper{\gambund}{0}{\frmbr{e}}\,.
\end{equation*}
For $t=1$ we will write just $\toper{\gambund}{0}{\frmbr{e}}$.

Suppose now that for \emph{each} frame $e$ we have a homotopy $G_{\frmbr{e}}(t)$ of connection coefficients such that $G_{\frmbr{e}}(0)$ is equal to coefficients of $\connbund$ in $e$ and $G_{\frmbr{e}}(1)=0$. For $\widetilde{e}=e g^{-1}$ let us relate $\ttoper{g\gambund g^{-1}+ g \rmd g^{-1}}{0}{\frmbr{\widetilde{e}}}$ and $\ttoper{\gambund}{0}{\frmbr{e}}$. Using (\ref{triv_tcovar}) one calculates
\begin{eqnarray}
\nonumber
\ttoper{g\gambund g^{-1}+ g \rmd g^{-1}}{0}{\frmbr{\widetilde{e}}}(a_{\frmbr{\widetilde{e}}})
&=&g \circ
\ttoper{\gambund}{g^{-1}\rmd g}{\frmbr{e}}(a_{\frmbr{e}})
\circ g^{-1}\\
&=&g \circ
\ttoper{\gambund}{g^{-1}\rmd g}{\frmbr{e}}
\left(\ttinvoper{\gambund}{0}{\frmbr{e}}\left(\ttoper{\gambund}{0}{\frmbr{e}}(a_{\frmbr{e}})\right)\right)
\circ g^{-1}.
\label{triv_T_covariance}
\end{eqnarray}
It is a matter of straightforward observation that for $t=1$ the operator 
\begin{equation*}
g \circ
\toper{\gambund}{g^{-1}\rmd g}{\frmbr{e}}
\left(\tinvoper{\gambund}{0}{\frmbr{e}}\left(\cdot\right)\right)
\circ g^{-1}
\end{equation*}
transports matrices representing flat sections belonging to $\wbund_{D_S}$ of frame $e$ to matrices belonging to $\wbund_{D_S}$ of frame $\widetilde{e}$. By means of $U(t)$ defined in (\ref{triv_U_def}) one may rewrite (\ref{triv_T_covariance}) at $t=1$ as
\begin{multline*}
\toper{g\gambund g^{-1}+ g \rmd g^{-1}}{0}{\frmbr{\widetilde{e}}}(a_{\frmbr{\widetilde{e}}})
=\\
g \circ
\uinvoper{\gambund}{g^{-1}\rmd g}{\frmbr{e}} \circ
\uoper{\gambund}{0}{\frmbr{e}} 
\circ\,
\toper{\gambund}{0}{\frmbr{e}}(a_{\frmbr{e}}) 
\circ
\uinvoper{\gambund}{0}{\frmbr{e}}
\circ\,
\uoper{\gambund}{g^{-1}\rmd g}{\frmbr{e}} \circ
g^{-1},
\end{multline*}
with $U=U(1)$. Then, putting $V_{\frmbr{e}}(g,\gambund)=\uinvoper{\gambund}{0}{\frmbr{e}} \circ \uoper{\gambund}{g^{-1}\rmd g}{\frmbr{e}} \circ g^{-1}$, one obtains
\begin{equation}
\label{triv_wbund_covar}
\toper{g\gambund g^{-1}+ g \rmd g^{-1}}{0}{\frmbr{\tilde{e}}}(a_{\frmbr{\tilde{e}}})=
V_{\frmbr{e}}^{-1}(g,\gambund)
\circ
\toper{\gambund}{0}{\frmbr{e}}(a_{\frmbr{e}}) \circ
V_{\frmbr{e}}(g,\gambund).
\end{equation}
\begin{lmm}
$V_{\frmbr{e}}(g,\gambund)$ is a flat section belonging to $\wbund_{D_S}$ of the frame $e$.
\end{lmm}
\begin{proof}
We are going to show that $D_S V(g,\gambund)=0$. Let us compute $V^{-1}(g,\gambund) \circ D_S V(g,\gambund)$. (All calculations are performed in the frame $e$. We omit subscript ${}_{\frmbr{e}}$ within this proof).
\begin{eqnarray}
\nonumber
V^{-1}(g,\gambund) &\circ& D_S V(g,\gambund)\\ 
\nonumber
&=&
g \circ
\uinvoper{\gambund}{g^{-1}\rmd g}{  } \circ
\uoper{\gambund}{0}{  } \circ
D_S \left(
\uinvoper{\gambund}{0}{  } \circ
\uoper{\gambund}{g^{-1}\rmd g}{  } \circ
g^{-1}\right)\\
\nonumber
&=&
g \circ \left(
\uinvoper{\gambund}{g^{-1}\rmd g}{  } \circ
\uoper{\gambund}{0}{  } \circ
D_S \left( 
\uinvoper{\gambund}{0}{  } \right) \circ
\uoper{\gambund}{g^{-1}\rmd g}{  } 
\right.\\
& &+\left.
\uinvoper{\gambund}{g^{-1}\rmd g}{  } \circ
D_S \uoper{\gambund}{g^{-1}\rmd g}{  } 
\right) \circ g^{-1} + g\circ D_S g^{-1}
\label{triv_flatv_1}
\end{eqnarray}
Now, analyze $\uoper{\gambund}{0}{  } \circ
D_S \uinvoper{\gambund}{0}{  }$. One can construct connection
\begin{equation*}
\smiled{D}=\rmd + \frac{\rmi}{h} [\smiled{\gamma},\cdot \,]=D_S +  [\uoper{\gambund}{0}{  } \circ D_S \uinvoper{\gambund}{0}{  }, \cdot \,]\,,
\end{equation*}
with
\begin{equation}
\label{triv_flatv_2}
\smiled{\gamma}=\gamma_S -\rmi h \uoper{\gambund}{0}{  } \circ D_S \uinvoper{\gambund}{0}{  }=
\uoper{\gambund}{0}{  } \circ \gamma_S \circ \uinvoper{\gambund}{0}{  } -\rmi h \uoper{\gambund}{0}{  } \circ \rmd	 \uinvoper{\gambund}{0}{  } \, .
\end{equation}
From lemma \ref{triv_lmm_gamma} it follows that
\begin{equation}
\label{triv_flatv_3}
\gamma_S=\uinvoper{\gambund}{0}{  } \circ \gamma \circ \uoper{\gambund}{0}{  } -\rmi h \uinvoper{\gambund}{0}{  } \circ \rmd \uoper{\gambund}{0}{  } \, ,
\end{equation}
where $\gamma$ comes from our initial connection $D=\rmd +\frac{\rmi}{h}[\gamma,\cdot,\,]$ in $\wbund$. Inserting (\ref{triv_flatv_3}) into (\ref{triv_flatv_2}) we get $\smiled{\gamma}=\gamma$.	Using this result again in (\ref{triv_flatv_2}) one obtains
\begin{equation*}
\uoper{\gambund}{0}{  } \circ D_S \uinvoper{\gambund}{0}{  }=\frac{\rmi}{h}(\gamma - \gamma_S).
\end{equation*}
When substituted into (\ref{triv_flatv_1}) this yields
\begin{eqnarray*}
V^{-1}(g,\gambund) \circ D_S V(g,\gambund)
&=&
g \circ \left(
\uinvoper{\gambund}{g^{-1}\rmd g}{  } \circ
\frac{\rmi}{h}(\gamma - \gamma_S)
\circ
\uoper{\gambund}{g^{-1}\rmd g}{  } 
\right.\\
& &+\left.
\uinvoper{\gambund}{g^{-1}\rmd g}{  } \circ
D_S \uoper{\gambund}{g^{-1}\rmd g}{  } 
\right) \circ g^{-1} + g\circ D_S g^{-1}.
\end{eqnarray*}
After application of $D_S=\rmd +\frac{\rmi}{h}[\gamma_S,\cdot\,]$ this formula turns into
\begin{eqnarray*}
V^{-1}(g,\gambund) \circ D_S V(g,\gambund)
&=&
g \circ \left(
\uinvoper{\gambund}{g^{-1}\rmd g}{  } \circ
D\uoper{\gambund}{g^{-1}\rmd g}{  } 
+
\frac{\rmi}{h}(\gamma - \gamma_S)
\right) \circ g^{-1} \\ 
& &+g\circ D_S g^{-1}.
\end{eqnarray*}
Let $\gamma_{g^{-1}\rmd g}$ be $\gamma_S$ of $\widetilde{e}$ transformed to $e$. Using lemma \ref{triv_lmm_gamma} one quickly computes
\begin{equation*}
\uinvoper{\gambund}{g^{-1}\rmd g}{  } \circ
D\uoper{\gambund}{g^{-1}\rmd g}{  }=\frac{\rmi}{h}(\gamma_{g^{-1}\rmd g} - \gamma).
\end{equation*} 
Thus
\begin{equation*}
V^{-1}(g,\gambund) \circ D_S V(g,\gambund)
=
g \circ \frac{\rmi}{h} \left(
\gamma_{g^{-1}\rmd g} - \gamma_S
\right) \circ g^{-1} + g\circ D_S g^{-1}.
\end{equation*}
Now $\gamma_S=\omega_{ij}y^i \rmd x^j + 1/2 \Gamma_{ijk}y^i y^j \rmd x^k + r_S$ and $\gamma_{g^{-1}\rmd g}=\omega_{ij}y^i \rmd x^j + 1/2 \Gamma_{ijk}y^i y^j \rmd x^k -\rmi h g^{-1}\rmd g+ r_{g^{-1}\rmd g}$. Repeating considerations of remark \ref{fedo_rmrk_flatbund} we find that $r_S=r_{g^{-1}\rmd g}$. Hence
\begin{equation*}
V^{-1}(g,\gambund) \circ D_S V(g,\gambund)
=
\rmd g \circ g^{-1} + g\circ D_S g^{-1}.
\end{equation*}
Furthermore, $D_S g^{-1}=\rmd g^{-1}$ because there are no noncentral endomorphisms in $\gamma_S$  and $g^{-1}$ contains no $y^i$. Thus, we finally arrive at the formula
$V^{-1}(g,\gambund) \circ D_S V(g,\gambund)=0$ and consequently $D_S V(g,\gambund)=0$.
\end{proof}

Making use of above lemma, we can transport relation (\ref{triv_wbund_covar}) back to $C^{\infty}(\ebund)[[h]]$. Let 
\begin{equation*}
\moper{\gambund}{\frmbr{e}}(F_{\frmbr{e}}) := Q_S^{-1} (\toper{\gambund}{0}{\frmbr{e}} (Q (F_{\frmbr{e}})))
\end{equation*}
for $F$ being a local section of $\ebund$ represented in the frame $e$ by the matrix $F_{\frmbr{e}}$. Notice that $\moper{\gambund}{\frmbr{e}}$ is obviously equivalence of star products generated by $D$ and $D_S$.
Define 
\begin{equation}
\label{triv_nonc_g}
\widehat{g}_{\frmbr{e}}\left(g,\gambund\right) := Q_S^{-1}(V_{\frmbr{e}}^{-1}(g,\gambund))
=Q_S^{-1}(g \circ
\uinvoper{\gambund}{g^{-1}\rmd g}{\frmbr{e}} \circ
\uoper{\gambund}{0}{\frmbr{e}}).
\end{equation}
Then, (\ref{triv_wbund_covar}) yields
\begin{equation}
\label{triv_ebund_covar}
\moper{g\gambund g^{-1}+ g \rmd g^{-1}}{\frmbr{\tilde{e}}}(B_{\frmbr{\tilde{e}}})=
\widehat{g}_{\frmbr{e}}\left(g,\gambund\right)
*_S
\moper{\gambund}{\frmbr{e}}(B_{\frmbr{e}})
*_S
\widehat{g}_{\frmbr{e}}^{-1}\left(g,\gambund\right).
\end{equation}
Using (\ref{triv_nonc_g}) and (\ref{triv_tcovar}) one may observe, that the relation 
\begin{equation}
\label{triv_g_consistency}
\widehat{g}_{\frmbr{e}}(g'g,\gambund)=\widehat{g}_{\frmbr{\tilde{e}}}(g',g\gambund g^{-1}+g \rmd g^{-1})*_S \widehat{g}_{\frmbr{e}}(g,\gambund)
\end{equation}
holds, providing "consistency condition" analogous to that described in \cite{jurco0, schupp}.
\subsection{Trivialization of $\wbund_{D_S}$ to $\wbund_{D_T}$}
The next ingredient required for building "noncommutative connection coefficients" $\widehat{\Gamma}$ and "curvature" $\widehat{R}$ is a noncommutative version of $\frac{\partial}{\partial x^i}$. We are going to use derivations $X_{i}$ defined in \cite{dobrski}. Their most important property is that they commute and thus we may define curvature in the usual way $\widehat{R}_{ij}=X_i(\widehat{\Gamma}_j)-X_j(\widehat{\Gamma}_i)+[\widehat{\Gamma}_i\stackrel{*}{,}\widehat{\Gamma}_j]$. In this subsection we briefly summarize results of \cite{dobrski} in the current context and also fix some notations for further pourposes. 

As stated in theorem \ref{triv_trhm_triv}, local isomorphism between $\wbund_{D_S}$ and $\wbund_{D_T}$ can be established by means of theorem \ref{triv_trhm_liouv}. Indeed, we may locally (for chosen Darboux coordinates)  fix some homotopy of symplectic connection coefficients $G^{\Sigma}_{ijk}(t)$ such that $G^{\Sigma}_{ijk}(0) \equiv 0$ and $G^{\Sigma}_{ijk}(1)=\Gamma_{ijk}$. Restricting our analysis to the case of constant normalization $\mu \equiv 0$ we obtain the family $D^{\Sigma}_t$ of Abelian connections with $D^{\Sigma}_0=D_T$ and $D^{\Sigma}_1=D_S$. The Hamiltonian compatible with theorem \ref{triv_trhm_liouv} is given by
\begin{equation*}
H^{\Sigma}(t)=-\frac{1}{6}Q^{\Sigma}_t\left(\frac{\rmd G^{\Sigma}_{ijk}}{\rmd t}y^i y^j y^k\right)
\end{equation*}
and thus, we obtain desired isomorphism $T_{\Sigma}$ transporting $\wbund_{D_T}$ to $\wbund_{D_S}$. It is quite natural to choose 
\begin{equation}
\label{wdt_wds_hmtp}
G^{\Sigma}_{ijk}(t)=f(t) \Gamma_{ijk}\, ,
\end{equation}
where $f : [0,1] \to \mathbb{R}$ satisfies $f(0)=1$ and $f(1)=0$. In this case one can calculate, that up to $h^2$ the isomorphism $T_{\Sigma}$ is independent of any particular choice of $f(t)$.

Using inverse of $T_{\Sigma}$ we may define 
\begin{equation*}
\lambda_i := \omega_{ij} Q_S^{-1} T_{\Sigma}^{-1} Q_T x^j,
\end{equation*}
with the commutation property 
\begin{equation}
\label{wdt_wds_lambda}
\frac{i}{h}[\lambda_i \stackrel{*_S}{,} \lambda_j] = -\omega_{ij}.
\end{equation}
We also introduce derivations
\begin{equation}
\label{wdt_wds_xder}
X_i f := \frac{i}{h}[\lambda_i \stackrel{*_S}{,} f].
\end{equation}
Using (\ref{wdt_wds_lambda}) it can be easily deduced that  $X_i X_j f= X_j X_i f$. Obviously, if one works with a flat symplectic connection and locally $\Gamma_{ijk} \equiv 0$, then $T_{\Sigma}=\id$ for homotopies of type (\ref{wdt_wds_hmtp}), and $X_i$ reduces to $\frac{\partial}{\partial x^i}$. In general the following formulae hold \cite{dobrski}:
\begin{equation*}
\lambda_i=\omega_{ij}x^j-\frac{h^2}{48}\frac{\partial \Gamma_{jkl}}{\partial x^i}\Gamma^{jkl}+O(h^3)
\end{equation*}
and
\begin{eqnarray*}
 \nonumber X_i f
&=&\frac{\partial f}{\partial x^i}
-h^2\left\{
\frac{1}{48} \omega^{ls} \frac{\partial f}{\partial x^s}  \frac{\partial}{\partial x^i}\left(\frac{\partial\Gamma_{mjk}}{\partial x^l} \Gamma^{mjk}\right)
+ \frac{1}{16} \omega^{ls} \frac{\partial^2 f}{\partial x^s \partial x^k} \frac{\partial (\Gamma^{mjk} \Gamma_{mjl})}{\partial x^i} \right.\\ 
& &
+\left.\frac{1}{24}\frac{\partial^3 f}{\partial x^m \partial x^j \partial x^k} \frac{\partial \Gamma^{mjk}}{\partial x^i}
\right\}+O(h^3)\, .
\end{eqnarray*}
\section{Seiberg-Witten Equations}
We are ready to show how Seiberg-Witten equations and their solutions can be obtained from the Fedosov formalism. Following ideas of Jur\v{c}o and Schupp \cite{jurco1} let us (for some fixed Darboux coordinates) introduce "noncommutative connection coefficients" $\widehat{\Gamma}$ according to the rule
\begin{equation}
\label{sw_nonc_gamma}
\widehat{\Gamma}_i (\gambund)_{\frmbr{e}} := \frac{\rmi}{h}\left(\moper{\gambund}{\frmbr{e}}(\omega_{ij} x^j) - \lambda_i\right).
\end{equation}
Using relations (\ref{triv_ebund_covar}) and (\ref{wdt_wds_xder}) one easily computes 
\begin{eqnarray}
\nonumber
\widehat{\Gamma}_i (g \gambund g^{-1} + g \rmd g^{-1})_{\frmbr{\widetilde{e}}} &=&  \frac{i}{h}\left(\moper{g \gambund g^{-1} + g \rmd g^{-1}}{\frmbr{\widetilde{e}}}(\omega_{ij} x^j) - \lambda_i \right)\\
\nonumber
&=&
\frac{i}{h} \left(\widehat{g}\left(g,\gambund\right) *_S
\moper{\gambund}{\frmbr{e}}(\omega_{ij} x^j) *_S
\widehat{g}^{-1}\left(g,\gambund\right) -\lambda_i \right)
\\
\nonumber
&=&\widehat{g}\left(g,\gambund\right) *_S \frac{i}{h} \left(
\moper{\gambund}{\frmbr{e}}(\omega_{ij} x^j) -\lambda_i \right)*_S
\widehat{g}^{-1}\left(g,\gambund\right) \\
\nonumber
& &+ 
\frac{i}{h} \widehat{g}\left(g,\gambund\right) *_S [\lambda_i \stackrel{*_S}{,}\widehat{g}^{-1}\left(g,\gambund\right)]\\
\nonumber
&=&\widehat{g}\left(g,\gambund\right) *_S \widehat{\Gamma}_i (\gambund)_{\frmbr{e}} *_S
\widehat{g}^{-1}\left(g,\gambund\right)\\ 
& &+ \widehat{g}\left(g,\gambund\right) *_S X_i\left(\widehat{g}^{-1}\left(g,\gambund\right)\right).
\label{sw_ncgamma_covar}
\end{eqnarray}
Thus, we have arrived at the formula which reduces to the well known Seiberg-Witten equations if $*_S$ is the Moyal product (case of $\Gamma_{ijk} \equiv 0$). The important remark which should be made here is that one does not need to \emph{solve} these equations to obtain $\widehat{\Gamma}$ and $\widehat{g}$. Instead, noncommutative gauge objects can be \emph{computed} with recursive methods of Fedosov equivalence theory. Also let us notice that in this way one obtains solutions of Seiberg-Witten equations for arbitrary gauge group (since Fedosov deformation quantization of $\ebund$ is valid for arbitrary $\bund$).

The "noncommutative curvature" may be defined in quite usual form
\begin{equation*}
\widehat{R}_{ij\,\frmbr{e}} := X_i(\widehat{\Gamma}_{j\,\frmbr{e}})-X_j(\widehat{\Gamma}_{i\,\frmbr{e}})+[\widehat{\Gamma}_{i\,{\frmbr{e}}}\stackrel{*_S}{,}\widehat{\Gamma}_{j\,{\frmbr{e}}}].
\end{equation*}
Simple calculation yields
\begin{eqnarray*}
\widehat{R}_{ij\,\frmbr{e}} &=& 
-\frac{1}{h^2}\left([\moper{\gambund}{\frmbr{e}}(\omega_{ik} x^k)\stackrel{*_S}{,}\moper{\gambund}{\frmbr{e}}(\omega_{jl} x^l)] -
[\lambda_i \stackrel{*_S}{,} \lambda_j] \right) \\
&=&-\frac{1}{h^2}\left(\moper{\gambund}{\frmbr{e}}\left([\omega_{ik} x^k\stackrel{*}{,}\omega_{jl} x^l]\right) +
\omega_{ij} \right),
\end{eqnarray*}
where relation (\ref{wdt_wds_lambda}) and equivalence $M$ of $*$ (star product in $C^{\infty}(\ebund)[[h]]$ generated by $\connbund$ and $\connsymp$) and $*_S$ have been used. With this formula one quickly computes
\begin{equation*}
\widehat{R}_{ij\,\frmbr{\widetilde{e}}} = \widehat{g} *_S \widehat{R}_{ij\,\frmbr{e}} *_S\widehat{g}^{-1}
\end{equation*}
since $\omega_{ij}$ is constant in Darboux coordinates.

To obtain some kind of Bianchi identities let us define
\begin{equation*}
\widehat{D}_i \widehat{R}_{jk} := X_i (\widehat{R}_{jk}) + [\widehat{\Gamma}_{i}\stackrel{*_S}{,}\widehat{R}_{jk}].
\end{equation*}
Then the formula
\begin{equation}
\label{sw_bianchi}
\widehat{D}_{[i} \widehat{R}_{jk]} \equiv 0
\end{equation}
holds as a purely algebraic consequence of $X_i X_j = X_j X_i$ and the Jacobi identity. Alternatively, one may calculate that
\begin{equation*}
\widehat{D}_i \widehat{R}_{jk\,\frmbr{e}} = -\frac{\rmi}{h^3} \moper{\gambund}{\frmbr{e}}\left([\omega_{il} x^l\stackrel{*}{,}[\omega_{jm} x^m \stackrel{*}{,} \omega_{kp} x^p]]\right)
\end{equation*}
and apply antisymmetrization to get (\ref{sw_bianchi}).

Explicit formulae for noncommutative gauge objects require calculation of explicit form of $\toper{\gambund}{0}{\frmbr{e}}$, $\uinvoper{\gambund}{0}{\frmbr{e}}$ and $\uoper{\gambund}{g^{-1}\rmd g}{\frmbr{e}}$. This is quite tedious but otherwise simple task, involving only recursive computations (formulae (\ref{fedo_abeliter}), (\ref{fedo_q_iter}), (\ref{triv_ebund_ham}), (\ref{triv_heis_int}) and (\ref{triv_U_def})) with known initial points for recurrences. One obtains (from now on we omit ${}_{\frmbr{e}}$ subscript)
\begin{eqnarray}
\nonumber
\toper{\gambund}{0}{}(Q(\omega_{ij}x^j))&=&
Q_S\bigg(
\omega_{ij}x^j - ih \gambund_i 
+\frac{h^2}{2} \omega^{sr} \int_{0}^{1}
\bigg\{
\dotghmtp{\gambund}{0}{}{s}(\tau) , 
\int_{0}^{\tau} \partial^{(\eta)}_{(i} \dotghmtp{\gambund}{0}{}{r)}(\eta)\\ 
\nonumber
& & + \left[\dotghmtp{\gambund}{0}{}{r}(\eta),\ghmtp{\gambund}{0}{}{i}(\eta) - \gambund_i\right]
\rmd \eta -\curvbund_{ri}(\tau)
\bigg\} \rmd \tau 
+ h^2 \mu_{5i}\\
\nonumber
& &+ \frac{h^2}{2} \tensor{\Gamma}{^{sr}_i} \int^{0}_{1} 
\left\{ \dotghmtp{\gambund}{0}{}{s}(\tau), \ghmtp{\gambund}{0}{}{r}(\tau) - \gambund_r \right\}
\rmd \tau \\
& &+\frac{h^2}{4} \omega^{sr} \left\{ \gambund_s , \curvbund_{ri}(0) \right\}
+O(h^3)\bigg)
\label{sw_triv_explicit}
\end{eqnarray}
and
\begin{eqnarray}
\nonumber
\widehat{\Gamma}_i (\gambund) &=& \gambund_{i} 
+\frac{\rmi h}{2} \omega^{sr} \int_{0}^{1}
\bigg\{
\dotghmtp{\gambund}{0}{}{s}(\tau) , 
\int_{0}^{\tau} \partial^{(\eta)}_{(i} \dotghmtp{\gambund}{0}{}{r)}(\eta) \\
\nonumber
& &+\left[\dotghmtp{\gambund}{0}{}{r}(\eta),\ghmtp{\gambund}{0}{}{i}(\eta) - \gambund_i\right]
\rmd \eta -\curvbund_{ri}(\tau) \bigg\} \rmd \tau 
+ \rmi h \mu_{5i}\\
\nonumber
& &+ \frac{\rmi h}{2} \tensor{\Gamma}{^{sr}_i} \int^{0}_{1} 
\left\{ \dotghmtp{\gambund}{0}{}{s}(\tau), \ghmtp{\gambund}{0}{}{r}(\tau) - \gambund_r \right\}
\rmd \tau \\
& &+\frac{\rmi h}{4} \omega^{sr} \left\{ \gambund_s , \curvbund_{ri}(0) \right\} 
+\frac{\rmi h}{48}\frac{\partial \Gamma_{jkl}}{\partial x^i}\Gamma^{jkl}+O(h^2),
\label{sw_gammgen_explicit}
\end{eqnarray}
where $\{\,\cdot,\cdot \,\}$ denotes anticommutator, $\curvbund(t)$ stands for the family of curvatures generated by  $\ghmtp{\gambund}{0}{\, }{} (t)$, $\partial^{(t)}_i A_j=\frac{\partial}{\partial x^i} A_j - \tensor{\Gamma}{^{k}_{ji}}A_k+[\ghmtp{\gambund}{0}{}{i}(t),A_j]$ and $\mu_{5i}$ comes from $\mu$ due to relation  $\mu= \cdots + h^2 \mu_{5i}y^i + \cdots$. (For sake of simplicity we have restricted ourselves to the case of $\deg m(t), \, \deg \mu \geq 4$ when arriving at above formula). Natural choice of homotopy $G(t)$ in the form $\ghmtp{\gambund}{0}{\, }{} (t)= f(t) \gambund$, where $f: [0,1] \to \mathbb{R}$, $f(0)=1$, $f(1)=0$, leads to
\begin{subequations}
\label{sw_expl}
\begin{eqnarray}
\nonumber
\widehat{\Gamma}_i (\gambund) &=& \gambund_{i} + \rmi h \left( \frac{1}{4} \omega^{jk} \left\{\gambund_j, \curvbund_{ki} + \frac{\partial \gambund_i}{\partial x^k} \right\} +  \mu_{5i}
+\frac{1}{48}\frac{\partial \Gamma_{jkl}}{\partial x^i}\Gamma^{jkl}
\right)\\
& &+O(h^2),\\
\widehat{g}(g,\gambund) &=& g+ \frac{\rmi h}{4} g \,\omega^{jk} \left( \frac{\partial g^{-1}}{\partial x^j} \frac{\partial g}{\partial x^k} + \left\{ \gambund_j, g^{-1} \frac{\partial g}{\partial x^k} \right\} \right) + O(h^2).
\end{eqnarray}
\end{subequations}
After switching to field theory conventions for gauge objects and to the most frequently used convention for Moyal product i.e. 
\begin{equation*}
\gambund = -iA, \qquad \curvbund = -i F, \qquad \widehat{\Gamma} = -i \widehat{A}, \qquad \tilde{h}=-h
\end{equation*}
relations (\ref{sw_expl}) may be rewritten as
\begin{subequations}
\begin{eqnarray}
\nonumber
\widehat{A}_i (A) &=& A_{i} + \tilde{h} \left( - \frac{1}{4} \omega^{jk} \left\{A_j, F_{ki} + \frac{\partial A_i}{\partial x^k} \right\} +  \mu_{5i}
+ \frac{1}{48}\frac{\partial \Gamma_{jkl}}{\partial x^i}\Gamma^{jkl}
\right)\\
& &+O(h^2),\\
\widehat{g}(g,A) &=& g - \frac{\rmi \tilde{h}}{4} g \,\omega^{jk} \left( \frac{\partial g^{-1}}{\partial x^j} \frac{\partial g}{\partial x^k} - \rmi \left\{ A_j, g^{-1} \frac{\partial g}{\partial x^k} \right\} \right) + O(h^2).
\end{eqnarray}
\end{subequations}
For infinitesimal version of gauge transformations ($g=e^{\rmi \chi}$, $\widehat{g}=e^{\rmi \widehat{\chi}}$ with small $\chi$ and $\widehat{\chi}$) it comes that
\begin{equation}
\widehat{\chi}(\chi,A)=\chi+\frac{\tilde{h}}{4}\omega^{jk}\left\{\frac{\partial \chi}{\partial x^j}, A_k\right\}+O(h^2). 
\end{equation}
Thus, we have arrived at formulae, which for $\Gamma_{ijk} \equiv 0$ and $\mu_{5i} \equiv 0$ reduce to the well known solutions of Seiberg-Witten equations \cite{seibwitt}.
\section{Final Comments}
We have shown how Seiberg-Witten equations can be derived within framework of equivalence theory of Fedosov star products. In our approach their solutions may be obtained up to arbitrary order of $h$ by precisely defined recurrences. Let us summarize all steps of this procedure. The starting point is given by the bundles $\bund$, $\ebund$ with connections $\connbund$, $\connend$ respectively, and the Fedosov deformation quantization of $\ebund$. We are looking for noncommutative version of $\gambund$. At the very beginning one must fix homotopies $\ghmtp{\gambund}{0}{\frmbr{e}}{}(t)$ for each possible frame $e$. Then, lifting $Q(\omega_{ij} x^j)$ should be calculated. (This involves computation of $r$ with formula (\ref{fedo_abeliter})). Next, the Hamiltonian $H(t)$ is determined from (\ref{triv_ebund_ham}). Both $Q(\omega_{ij} x^j)$ and $H(t)$ are put into equation (\ref{triv_heis_int}) for iterative calculation of $\toper{\gambund}{0}{\frmbr{e}}(Q(\omega_{ij} x^j))$. Finally we project the result back to $C^{\infty}(\ebund)[[h]]$ and use $\lambda_i$ (which is calculated in quite similar way) to obtain $\widehat{\Gamma}_i$ according to (\ref{sw_nonc_gamma}). For $\widehat{g}$ the equation 
(\ref{triv_U_def}) must be used twice. First for computation of $\uinvoper{\gambund}{0}{\frmbr{e}}$, and then for $\uoper{\gambund}{g^{-1}\rmd g}{\frmbr{e}}$. For the second iteration one must use Hamiltonian calculated for the homotopy in the frame $\widetilde{e}=eg^{-1}$. Then equation (\ref{triv_nonc_g}) gives $\widehat{g}$.

Let us also mention that we have achieved some kind of generalization of Seiberg-Witten equations from Moyal to Fedosov type of noncommutativity. Indeed, in covariance relation (\ref{sw_ncgamma_covar}) the star product $*_S$ is just the matrix product with multiplication of elements replaced by Fedosov star product of functions. Notice, that in order to keep well defined curvature $\widehat{R}$ we have to introduce some commuting derivations. This is done by importing derivations $X_i$ defined in \cite{dobrski}.

Ambiguities in solutions characteristic to Seiberg-Witten equations arise in our setting in two contexts. First, as a consequence of ambiguity in choice of homotopies $\ghmtp{\gambund}{0}{\frmbr{e}}{}(t)$.  This is explicitly manifested in formulae (\ref{sw_triv_explicit}) and (\ref{sw_gammgen_explicit}).
The second source of possible ambiguities comes from internal degrees of freedom of Fedosov deformation quantization. We have (to some extent) analyzed consequences of arbitrariness in choice of
$\mu$, but it should be mentioned, that this is not the only parameter which can be varied in Fedosov formalism. Indeed, one may also change curvature of Abelian connections, which is fixed in this paper to be $\Omega=-1/2\, \omega_{ij}\rmd x^i \wedge \rmd x^j$. As in the case of nonconstant $\mu$, equivalences between star products with different $\Omega$ would yield some extra terms in solutions (\ref{sw_expl}).

Presented construction is (as in the case of the usual Seiberg-Witten equations) local and dependent on the choice of particular Darboux coordinates. However, it has been shown that Seiberg-Witten correspondence can be understood as a local trivialization of the global structure -- Fedosov quantization of endomorphism bundle. Thus, it is natural to pass from the local description of noncommutative field theory (by means of Seiberg-Witten map and Moyal product) to the global one (expressed in the language of Fedosov $*$-product of endomorphisms). The case of noncommutative general relativity will be covered in this manner in the author's forthcoming paper.

\section*{Acknowledgments}
I would like to thank professor Maciej Przanowski for reviewing initial version of this paper and helpful remarks.

\end{document}